\documentclass[]{article} 

\usepackage[margin=3cm]{geometry}
\usepackage[english]{babel}
\usepackage[latin1]{inputenc}
\usepackage{indentfirst}
\usepackage{amsmath, amssymb, mathrsfs, dsfont, mathtools, empheq}
\usepackage{amsthm}
\usepackage{graphicx}
\usepackage{comment}

\theoremstyle{plain}
\newtheorem{thm}{Theorem} [section]
\newtheorem{theorem}[thm]{Theorem}
\newtheorem{lemma}[thm]{Lemma}
\newtheorem{proposition}[thm]{Proposition}
\newtheorem{corollary}[thm]{Corollary}

\theoremstyle{definition}
\newtheorem{definition}[thm]{Definition}
\newtheorem{remark}[thm]{Remark}

\def\be{\begin{equation}}
\def\ee{\end{equation}}
\def\bc{\begin{center}}
\def\ec{\end{center}}

\begin{document}
\markboth{A. Barra, F. Guerra, E. Mingione}{Philosophical Magazine Letters}

\title{{\itshape Interpolating the Sherrington-Kirkpatrick replica trick}
}
\author{Adriano Barra$^{\rm a}$$^{\ast},$\thanks{$^\ast$Corresponding author. Email: adriano.barra@roma1.infn.it
\vspace{6pt}} Francesco Guerra$^{\rm a b}$ and Emanuele Mingione$^{\rm a}$  \\ \vspace{6pt}  $^{\rm a}${\em{Dipartimento di Fisica, Sapienza Universit$\grave{a}$ di Roma, Piazzale Aldo Moro 5, 00185}}\\
$^{\rm b}${\em {Istituto Nazionale di Fisica Nucleare, Sezione di Roma}}\\
{\em{ Roma, Italy}}
}

\maketitle

\begin{abstract}
The interpolation techniques have become, in the past decades, a powerful approach
to lighten several properties of spin glasses within a simple mathematical framework.
Intrinsically, for their construction, these schemes were naturally implemented into the cavity field technique, or its variants as the stochastic stability or the random overlap structures.
\newline
However the first and most famous approach to mean field statistical mechanics with quenched disorder is the replica trick.
\newline
Among the models where these methods have been used (namely, dealing with frustration and complexity),
probably the best known is the Sherrington-Kirkpatrick spin glass:
\newline
In this paper we are pleased to apply the interpolation scheme
to the replica trick framework and test it directly to the cited paradigmatic model: interestingly this allows to obtain easily the replica-symmetric control and, synergically with the broken replica bounds, a description of the full RSB scenario, both coupled with several minor theorems. Furthermore, by treating the amount of replicas $n\in(0,1]$ as an interpolating parameter (far from its original interpretation) this can be though of as a quenching temperature close to the one introduce in off-equilibrium approaches and, within this viewpoint, the proof of the attended commutativity of the zero replica and the infinite volume limits can be obtained.\\

\textbf{Keywords}: Cavity Method, Spin Glasses, Replica Trick.

\end{abstract}

\section{Introduction}

Born as a sideline in the condensed matter division of modern theoretical physics,
spin glasses became soon the "harmonic oscillators"\footnote{We learn this beautiful metaphor by Ton Coolen,
that we thank.} of the new paradigm of complexity: hundreds -if not thousands- of papers developed from (and on)
this seminal model. Frustration, replica symmetry breaking, rough valleys of free energy, slow relaxational dynamics,
aging and rejuvenation (and much more) paved the mathematical and physical strands of a new approach to Nature,
where the protagonists are no longer the subjects by themselves but mainly the ways they interact. As a result, complex statistical
mechanics is invading areas far beyond condensed matter physics,
ranging from  biology (e.g. neurology  \cite{amit,hopf1,ton} and immunology  \cite{immuno,parisi})
to human sciences (e.g. sociology  \cite{CG,socio} or economics  \cite{bouchaud,coolen}) and much more (see  \cite{science} for instance).
\newline
Despite a crucial role has been played surely by the underlying graph theory
(due to breakthroughs obtained even there, i.e. with the understanding of the small worlds  \cite{watts}
or the scale free networks  \cite{barabasi}),
we would like to confer to the Sherrington-Kirkpatrick model -SK from now on- (or its concrete variants on graphs, as the Viana-Bray model  \cite{VB,GT2} just to cite one) a crucial role in this new science of complexity.
\newline
Among the methods developed for solving its thermodynamics  \cite{bovierbook,challenge}, the interpolation techniques, even though not yet so strong to solve the problem in fully autonomy, covered soon a key role
to -at least- lighten several properties of this system, working as a synergic alternative to the replica trick \cite{sk,sk2,MPV}, which is actually the first and most famous approach to mean field statistical mechanics with quenched disorder:
In fact, the interpolation scheme has been "naturally" implemented into the cavity field technique \cite{barra1,CLT,quadratic}, or its variants as the stochastic stability \cite{hopf1,contucci,ac} or the random overlap structures \cite{ass,arguin2}.
\newline
In this paper we want to study this model by extending the interpolating scheme, from the original cavity perspective to the replica trick: To allow this procedure we completely forget the original role played by the "amount" of replicas in the replica trick (tuned by a parameter $n\in(0,1]$) and think of it directly as a real  interpolating parameter. Interestingly this can intuitively though of as a quenching parameter coherently with its counterpart in the glassy dynamics (i.e. FDT violations  \cite{leticia1} \cite{leticia2}).  At first, once the mathematical strategy has been introduced in complete generality, we use it to
obtain a clear picture of the infinite volume and the zero replica limits at the replica symmetric level (by which
the whole original SK theory is reproduced), then, within the Parisi full replica symmetry breaking scenario, coupled with the broken replica bounds  \cite{g3}, other robustness
properties dealing with the exchange of these two limits are achieved as well.
\newline
The paper is therefore structured as follows:
\newline
In the next Section, $2$, we briefly introduce the model (and the ideas behind the replica trick strategy) while in Section $3$ we outline the strategy we want to apply to the model. All the other sections are then left to
the implementation of the interpolation into this framework and for presenting the consequent results.

\section{The Sherrington-Kirkpatrick mean field spin glass}

\subsection{The model and its related definitions}

The generic configuration of the Sherrington-Kirkpatrick model
 \cite{sk, sk2} is determined by the $N$ Ising variables
$\sigma_i=\pm1$, $i=1,2,\ldots,N$. The Hamiltonian of the model,
in some external magnetic field $h$, is
\begin{equation}
\label{SK} H_N(\sigma,h;J)=-\frac1{\sqrt N} \sum_{1 \leq i < j
\leq N} J_{ij} \sigma_i\sigma_j- h\sum_{1 \leq i \leq N} \sigma_i.
\end{equation}
The first term in (\ref{SK}) is a long range random two body
interaction, while the second represents the interaction of the
spins with the magnetic field $h$. The external quenched disorder
is given by the $N(N-1)/2$ independent and identically distributed
random variables $J_{ij}$, defined for each pair of sites. For the
sake of simplicity, denoting the average over this disorder by
$\mathbb{E}$, we assume each $J_{ij}$ to be a centered unit
Gaussian with averages
$$\mathbb{E}(J_{ij})=0,\quad \mathbb{E}(J_{ij}^2)=1.$$

For a given inverse temperature\footnote{Here and in the
following, we set the Boltzmann constant $k_{\rm B}$ equal to one,
so that $\beta=1/(k_{\rm B} T)=1/T$.} $\beta$, we introduce the
disorder dependent partition function $Z_{N}(\beta,h;J)$, the
quenched average of the free energy per site $f_{N}(\beta,h)$, the
associated averaged normalized log-partition function
$\alpha_N(\beta,h)$, and the disorder dependent Boltzmann-Gibbs
state $\omega$, according to the definitions
\begin{eqnarray}
\label{Z}
Z_N(\beta,h;J)&=&\sum_{\sigma}\exp(-\beta H_N(\sigma,h; J)),\\
\label{f}
-\beta f_N(\beta,h)&=& N^{-1} \mathbb{E}\ln Z_N(\beta,h)=\alpha_N(\beta,h),\\
\label{state}
\omega(A)&=&Z_N(\beta,h;J)^{-1}\sum_{\sigma}A(\sigma)\exp(-\beta
H_N(\sigma,h;J)),
\end{eqnarray}
where $A$ is a generic smooth function of $\sigma$.

Let us now introduce the important concept of replicas. We consider a
generic number $n$ of independent copies of the system,
characterized by the spin configurations $\sigma^{(1)}, \ldots ,
\sigma^{(n)}$, distributed according to the product state
$$
\Omega=\omega^{(1)} \times \omega^{(2)} \times \dots \times
 \omega^{(n)},
$$
where each $\omega^{(\alpha)}$ acts on the corresponding
$\sigma^{(\alpha)}_i$ variables, and all are subject to the {\em
same} sample $J$ of the external disorder.
\newline
The overlap  between two replicas $a,b$ is defined according to
\begin{equation}
\label{overlap} q_{ab}(\sigma^{(a)},\sigma^{(b)})={1\over N}
\sum_{1 \leq i \leq N}\sigma^{(a)}_i\sigma^{(b)}_i,
\end{equation}
and satisfies the obvious bounds $-1\le q_{ab}\le 1$.
\newline
For a generic smooth function $A$ of the spin configurations on
the $n$ replicas, we define the average $\langle A \rangle$ as
\begin{equation}\label{medie}
\langle A \rangle = \mathbb{E}\Omega
A\left(\sigma^{(1)},\sigma^{(2)},\ldots,\sigma^{(n)}\right),
\end{equation}
where the Boltzmann-Gibbs average $\Omega$ acts on the replicated
$\sigma$ variables and $\mathbb{E}$ denotes, as usual, the average
with respect to the quenched disorder $J$.

\subsection{The replica trick in a nutshell}

The replica trick consists in evaluating the logarithm of the
partition function through its power expansion, namely \be \log Z =
\lim_{n \to 0}\frac{Z^n-1}{n} \Rightarrow \langle \log Z \rangle =
\lim_{n \to 0}\frac{\langle Z^n \rangle - 1}{n}= \lim_{n \to
0}\frac1n \log\langle Z^n \rangle, \ee such that the (intensive) free energy can be written as
\be f_N(\beta,h) = \lim_{n \to 0} f_N(n,\beta,h),\ee
where  $f_N(n,\beta,h)$ is defined through
\be -\beta f_N(n,\beta,h) =  \alpha_N(n,\beta,h) =  \frac{1}{ N n } \log\langle Z^n \rangle.\ee
By assuming the validity of the following commutativity of the $n,N$ limits
\begin{equation}
\lim_{N \to \infty}\lim_{n \to 0}\alpha_N(n,\beta,h)=\lim_{n \to 0}\lim_{N \to \infty}\alpha_N(n,\beta,h)
\end{equation}
both Sherrington-Kirkpatrick (at the replica symmetric level  \cite{sk,sk2}) and Parisi (within the full RSB scenario
 \cite{parisi2,parisi3,parisi4}) gave a clear picture of the thermodynamics, which can be streamlined as follows:
At the replica symmetric level (i.e. by assuming replica equivalence, namely $q_{ab} = q$ for $a \neq b$, $1$ otherwise) we get
\be \alpha_{SK}(\beta)=\min_q\{\alpha(\beta,h,q)\}, \ee
where the trial function $\alpha(\beta,h,q)$ is defined as \be
\alpha(\beta,h,q)=\log2+\int d\mu(z)\log\cosh\Big(\beta (\sqrt{q}z + h) \Big)+\frac{\beta^2}{4}(1-q)^2.\ee
The selfconsistency relation for q reads off as
\be q_{SK}=\int d\mu(z)\tanh^2\Big(\beta(\sqrt{q_{SK}}z+h)\Big)).\ee
At the broken replica level we can write
\be \lim_{N \to \infty} \frac{1}{ N}\mathbb{E}\log{Z_{N}(\beta,J,h)}=\alpha(\beta,h)=-\beta f(\beta,h)=\alpha_P(\beta,h),\ee
where $\alpha_P(\beta,h)$, the fully broken replica solution, is defined as follows:
Let us consider the functional
\be \alpha_P(\beta,h,x)=\log 2 +f(0,y;x,\beta)\mid_{y=h}-\frac{\beta^2}{2}\int^1_0qx(q)dq,\ee
where $f(q,y;x,\beta)\equiv f(q,y)$ is solution of the equation
\be \partial_qf+\frac{1}{2}\partial^2_yf+\frac{1}{2}x(q)(\partial_yf)^2=0,\ee
with boundary $f(1,y)=\log\cosh(\beta y)$.  Then
\be \alpha_P(\beta,h)=\inf_{x\in{\cal X}}\alpha_P(\beta,h,x),\ee
where ${\cal X}$ is the convex space of the piecewise constant functions as introduced for instance in  \cite{g3}.

\section{The interpolating framework for the replica trick}

In this Section we present our strategy of investigation; namely we show some Theorems and Propositions
 whose implications will be exploited in the next Sections. For the sake of clearness we will omit
some straightforward demonstrations.
\newline
We want to think at the mapping among the one-replica and zero-replica as an interpolation scheme,
by the introduction of an auxiliary interpolating function, that we call $n$-quenched free energy, which (non trivially) bridges the system among $n=1$ and $n=0$, as
\begin{equation}
\varphi_N(n,\beta,h)= \frac{1}{Nn}\log{\mathbb{E}(Z^n_{N}(\beta,J,h))},
\end{equation}
where, for the sake of clearness $Z^n_{N}(\beta,J,h)\equiv{(Z_{N}(\beta,J,h))^n}$.
\newline
It is then worth stressing the next
\begin{theorem}\label{teorema1}
The following relation, among the interpolating function and the free energy, holds
\begin{equation}
\lim_{n \to 0}\varphi_N(n,\beta,h)=\alpha_N(\beta,h),
\end{equation}
furthermore
\begin{equation}
\varphi_N(n,\beta,h)\geq{\alpha_N(\beta,h)}
\end{equation}
for any  $n$.
\end{theorem}
\begin{proof} We can  expand in Taylor series in $n\in [0,1]$ to get
\begin{eqnarray}\nonumber
\log{\mathbb{E}(Z^n_{N}(\beta,J,h))} &=& 0+\mathbb{E}(\log{Z_{N}(\beta,J,h)})n+o(n^2)\Rightarrow \\
\lim_{n \to 0^+}\varphi_N(n,\beta,h) &=&\lim_{n \to 0}\frac{1}{Nn}(\mathbb{E}(\log{Z_{N}(\beta,J,h)})n+o(n^2))=\alpha_N(\beta,h).
\end{eqnarray}
The Jensen inequality ensures the second statement of the Theorem.
\end{proof}
\begin{proposition}\label{proposition1}
Through Theorem \ref{teorema1} we immediately obtain
\begin{equation}
\lim_{N \to \infty}\lim_{n \to 0}\varphi_N(n,\beta,h)=\alpha(\beta,h).
\end{equation}
\end{proposition}
We want to deepen now the properties of $\varphi_N(n,\beta,h)$ following the strategy outlined in  \cite{guerraspinglasses}:
\begin{proposition}\label{trepuntotre}
Let $i\in Q={\{1,...,N\}}$. Introduce positive weights $\forall i\longrightarrow{w_i}\in{\mathbb{R^+}}$. Let  $\forall i\longrightarrow{U_i}$ be a family of Gaussian random variables such that
$\mathbb{E}(U_i)=0$ and $\mathbb{E}({U_i}{U_j})=S_{ij}$, where $S_{ij}$ is a positive defined symmetric matrix.
\newline
For the functional $\varphi(n,t)=n^{-1}\log{\mathbb{E}(Z^n_t)}$,
where $Z_t=\sum_{i}w_i\exp({\sqrt{t}U_i})$, the following relation holds
\begin{equation}
\frac{d}{dt}\varphi(n,t)=\frac{1}{2}{\langle  S_{ii}\rangle }_n+\frac{(n-1)}{2}{\langle  S_{ij}\rangle }_n,
\end{equation}
\end{proposition}
where we introduced the following
\begin{definition}
$ {\langle A \rangle}_n=\mathbb{E}\Big(Z^n_t \mathbb{E}(Z^n_t)^{-1}\Omega(A)\Big)$ is a deformed state on
the $2$-product Boltzmann one, namely
$$\Omega(A)=\sum_{i,j}^{N}(Z^{-1}_tw_i\exp{\sqrt{t}U_i})(Z^{-1}_t\omega_j\exp{\sqrt{t}U_j})A, $$
where $A$ is an observable on $Q \times Q$,
$$\omega(A)=\sum_{i}^N(Z^{-1}_tw_i\exp{\sqrt{t}U_i})A,$$
being $A\in\mathcal{A}(Q)$.
\end{definition}
The following generalization, considering two families of random variables, can be easily obtained.
\begin{proposition}\label{trepuntocinque}
Let $i\in{Q=\{1,...,N\}}$ be a probability space and $\forall i\longrightarrow{w_i}\in{\mathbb{R^+}}$
be a probability weight and $\forall i\longrightarrow{U_i}$ a family of random Gaussian variables such that
$\mathbb{E}(U_i)=0$ and $\mathbb{E}({U_i}{U_j})=S_{ij}$, where $S_{ij}$  is a positive defined symmetric matrix.
\newline
Let $\forall i\longrightarrow{\tilde{U}_i}$ another family of random Gaussian variables such that
$\mathbb{E}(\tilde{U}_i)=0$ and $\mathbb{E}({\tilde{U}_i}{\tilde{U}_j})=\tilde{S}_{ij}$,
where $S_{ij}$ is a positive defined symmetric matrix. Let us further consider the functional
$\varphi(n,t)=n^{-1}\log{\mathbb{E}(Z^n_t)}$ (where $Z_t=\sum_{i}w_i\exp{(\sqrt{t}U_i+\sqrt{1-t}\tilde{U}_i)}$):
the following relation holds
\begin{equation}
\frac{d}{dt}\varphi(n,t)=\frac{1}{2}{\langle S_{ii}-\tilde{S}_{ii}\rangle }_n+
\frac{(n-1)}{2}{\langle  S_{ij}-\tilde{S}_{ij}\rangle }_n.
\end{equation}
\end{proposition}
We can then formulate the following
\begin{theorem}
If $\forall(i,j)\in Q\times Q$, $S_{ii}=\tilde{S}_{ii}$ and $S_{ij}\geq\tilde{S}_{ij}$, the following relation holds
$$\varphi(n,1)\leq\varphi(n,0), \ \ \forall n\in (0,1].$$
\end{theorem}
\begin{proof} Integrating among $0,1$ the functional we get
$\varphi(n,1)-\varphi(n,0)=\frac{1}{2}(n-1)\int^1_0dt{\langle  S_{ij}-\tilde{S}_{ij}\rangle }_n $,
whose r.h.s. is $\leq 0$ for $n\in (0,1]$.
\newline
Obviously the following relation tacitely holds: $\lim_{n \to 0}{\langle \cdot \rangle }_n=\langle \cdot \rangle$.
\end{proof}

Focusing on the Sherrington-Kirkpatrick model, as earlier introduced, and by using the results of the previous
Section, we still think at the $n$-variation as an interpolation and we can state the following
\begin{theorem}\label{teorema3}
Let us consider the functional $\psi_N(n,\beta,h)= n^{-1}\log{\mathbb{E}(Z^n_N(\beta,J,h))}=N\varphi_N(n,\beta,h)$:
$\psi_N(n,\beta,h)$ is super-additive in $N$,  $\forall n\in (0,1]$.
Furthermore
$$\lim_{N \to \infty}\varphi_N(n,\beta,h)=\sup_N \varphi_N(n,\beta,h)=\varphi(n,\beta,h), \ \textit{for any $n$}.$$
\end{theorem}
We omit the proof as it is analogous to the one achieved in  \cite{limterm}.
\begin{corollary}
Remembering that for super-additive (and bounded) functions we can write
\begin{equation}
\lim_{N \to \infty}\alpha_N(\beta,h)=\sup_N \alpha_N(\beta,h)=\alpha(\beta,h),
\end{equation}
we get a lower bound for $\varphi(n,\beta,h)$ as $\varphi(n,\beta,h)\geq \alpha(\beta,h)$ and
$ \sup_N \varphi_N(n,\beta,h)\geq \sup_N \alpha_N(\beta,h)$.
\end{corollary}

\section{Replica symmetric interpolation}

For the upper bound we have to tackle the replica symmetric approximation by using
a linearization strategy as follows\footnote{This procedure is deeply related to the mean field nature of the interactions, which ultimately allows to consider even the low temperature regimes as expressed in terms of high temperature solutions  \cite{t4}}: We introduce and define an interpolating partition function with $t \in [0,1]$ as
\begin{equation}
Z_t=\sum_{\{\sigma\}}\exp(\beta \widetilde{H}(t,\sigma))\exp\Big(\beta\ h\sum_{i}^N\sigma_i\Big),
\end{equation}
where, labeling with $K(\sigma)$ standard $\mathcal{N}(0,1)$ indexed by the configurations $\sigma$ and characterized by covariance $\mathbb{E}(K(\sigma)K(\sigma'))=q^2_{\sigma \sigma'}$ we defined
\begin{equation}
\widetilde{H}(t,\sigma)=\sqrt{t}\sqrt{\frac{N}{2}}K(\sigma)+\sqrt{1-t}\sqrt{q}\sum_{i}J_i\sigma_i,
\end{equation}
where $q$ will play the role of the replica-symmetric overlap,
and $J_i$ are random Gaussians i.i.d. $\mathcal{N}[0,1]$ independent also of $ K(\sigma)$ and such that
\be \mathbb{E}\Big((\beta\sqrt{q}\sum_{i}J_i\sigma_i)(\beta\sqrt{q}\sum_{j}J_j\sigma_j)\Big)=\beta^2Nqq_{\sigma\sigma'}.\ee
\begin{lemma}
Let us consider the functional $\varphi(t)=  (Nn)^{-1}\log{\mathbb{E}(Z^n_t)}$: We have that
\begin{eqnarray}
\varphi(1) &=& \frac{1}{Nn}\log{\mathbb{E}(Z^n_1)}=\varphi_N(n,\beta,h)\\
\varphi(0) &=& \log 2 +\frac{1}{n}\log{\int d\mu(z)\cosh^n\Big(\beta(\sqrt{q}z+h)\Big)}.
\end{eqnarray}
\end{lemma}
We are ready to state the next
\begin{theorem}\label{theorem4}
$\forall n\in (0,1]$ we have
\be \varphi_N(n,\beta,h)\leq \log2+\frac{1}{n}\log{\int d\mu(z)\cosh^n\Big(\beta(\sqrt{q}z+h)\Big)}+\frac{\beta^2}{4}(1-2q-(n-1)q^2)\ee
uniformly in $N$.
\end{theorem}
\begin{proof}
By applying Proposition \ref{trepuntocinque} we get
$$\frac{d}{dt}\varphi(t)=\frac{\beta^2}{4}-\frac{\beta^2}{2}q+\frac{(n-1)\beta^2}{4}{\langle q^2_{\sigma\sigma'}-2qq_{\sigma\sigma'}\rangle }_n,$$
then, completing with $q^2$ the square  at the r.h.s., and integrating back in $0,1$ we get the thesis.
\end{proof}
In complete analogy with the original SK theory we can define
\begin{eqnarray}
\alpha(n,\beta,h,q) &=& \log 2 +\frac{1}{n}\log{\int d\mu(z)\cosh^n\Big(\beta(\sqrt{q}z+h)\Big)}+\frac{\beta^2}{4}(1-2q-(n-1)q^2),\nonumber \\
\alpha_{RS}(n,\beta,h) &=& \min_q(\alpha(n,\beta,h,q)).
\end{eqnarray}
Then we get immediately the next
\begin{theorem}\label{theorem5}
$\forall n\in (0,1]$, $\varphi_N(n,\beta,h)\leq \alpha_{SK}(n,\beta,h)$ uniformly in $N$.
\end{theorem}
It is worth noting that the stationarity of $q$ becomes
\be\label{stazio} \frac{\partial}{\partial q}\alpha(n,\beta,h,q)=0\Rightarrow
 q_n=\frac{\int d\mu(z)\cosh^n\theta\tanh^2\theta}{\int d\mu(z)\cosh^n\theta}={\langle \tanh^2\theta\rangle}_n
 \ee
where we emphasized the $n$-dependence of $q$ via $q_n$, we used $\theta=\beta(\sqrt{q_n}z+h)$ for the sake of clearness, $d \mu$ as a standard Gaussian measure and the averages as
$${\langle F \rangle}_n=E\Big(\frac{Z^n}{\mathbb{E}(Z^n)}F\Big)=\frac{\int d\mu(z)\cosh^n\theta F}{\int d\mu(z)\cosh^n\theta}.$$
This ensures the validity of the next
\begin{theorem} For all the values of $n \in (0,1]$ we have
\begin{eqnarray}\nonumber
&& \alpha_{SK}(n,\beta,h)\geq\alpha_{SK}(\beta,h), \ \lim_{n \to 0}\alpha_{SK}(n,\beta,h)=\alpha_{SK}(\beta,h), \\
&& q_n\geq q_{SK}, \ \ \lim_{n \to 0}q_n=q_{SK}.\nonumber
\end{eqnarray}
\end{theorem}
Furthermore it is possible to show easily that, under specifical conditions, eq.(\ref{stazio}) defines a contraction, implicitly accounting for the high temperature regime\footnote{High temperature
is the $\beta$-region where there is only one solution, i.e. $q=0$, of the self-consistency relation:
When this condition breaks, phase transition to a broken replica phase appears; we label $\beta_c$ that particular value of the temperature.}. To this task we rewrite the latter as
\be q=\beta^2q\frac{\int d\theta\exp(-\frac{\theta^2}{2\beta^2q})\cosh^n\theta\tanh^2\theta}{\int d\theta\exp(-\frac{\theta^2}{2\beta^2q})\cosh^n(\theta)(\theta-n\beta^2q\tanh\theta)\theta},\ee
such that $\forall q \in \mathcal{R}\longrightarrow\|q\|\equiv|q|$.
\newline
Let us introduce the operator $\textbf{K}:q\longrightarrow\textbf{K}(q)$
defined via the original replica symmetric self-consistency relation and use for its norm $\|\textbf{K}\|\equiv \sup_q (\|\textbf{K}(q)\| / \|q\|)$. So we can state that
\begin{theorem}
$\exists (n,\beta):$ $\textbf{K}$ is a contraction in $\mathcal{R}$ and these
are related by $\beta_c(n)=\sqrt{1+n}^{-1}$: coherently with the previous results,
criticality is recovered at $\beta_c=1$ when $n \to 0$.
\end{theorem}
\begin{proof} By definition
$$\|\textbf{K}\|=\sup_q\Big\{\frac{\beta^2|q|}{|q|}\frac{|\int d\theta\exp(-\frac{\theta^2}{2\beta^2q})\cosh^n\theta\tanh^2\theta|}{|\int d\theta\exp(-\frac{\theta^2}{2\beta^2q})\cosh^n(\theta)(\theta-n\beta^2q\tanh\theta)\theta|}\Big\}.$$
By using the reversed triangular relation we get
$|\tanh\theta|\leq|\theta|\Rightarrow|\theta-n\beta^2q\tanh\theta|
\geq|(|\theta|-n\beta^2q|\tanh\theta|)|\geq|\theta||1-n\beta^2q|$ such that
\be
\|\textbf{K}\|\leq\sup_q\Big\{\frac{\beta^2}{|1-n\beta^2q|}\Big\}; \ \
q\in[0,1]\Rightarrow\|\textbf{K}\|\leq \frac{\beta^2}{|1-n\beta^2|}.\ee

So if $\beta^2\leq |1-n\beta^2|$, $\textbf{K}$ is a contraction and $q=0$ is the only solution of the
self consistency relation.
\end{proof}

\section{Broken replica interpolation}

To figure out an easy way to deal with the RSB scenario within an interpolating
framework, we  now rearrange the scaffold introduced in  \cite{guerraspinglasses} \cite{g3} as follows:
Beyond the structures outlines in Propositions \ref{trepuntotre},\ref{trepuntocinque}, we introduce $K\in \textbf{N}$
as an RSB-level counter such that, concretely,
$\forall (a,i)$ with $a=1...K$ and $i=1...N$ we use a family
$B^a_i$ of i.i.d. $\mathcal{N}[0,1]$, independent even by the $U_i$ and such that
\be\label{35} \mathbb{E}(B^a_iB^b_j)=\delta_{ab}\widetilde{S}^a{ij}.\ee
We introduce the averages with respect to the variables $B^K_i,B^{K-1}_i...B^1_i,U_i$ with the notation
$$ \mathbb{E}_a(\cdot)=\int d\mu(B^a_i)(\cdot)  \
\forall a=1...K, \ \mathbb{E}_0(\cdot)=\int d\mu(U_i)(\cdot),  \ \mathbb{E}(\cdot)=\mathbb{E}_0 \mathbb{E}_1... \mathbb{E}_K(\cdot),$$
and, $\forall n\in(0,1]$,  a family of order parameters $(m_1,...m_K)_n$ with $n<m_a<1\ \ \ \forall a=1,...,K$,  and
-recursively- the following r.v.
$$Z_K(t)=\sum_{i}w_i\exp{(\sqrt{t}U_i+\sqrt{1-t}\sum^K_{a=1}B^a_i)},
\ Z^{m_a}_{a-1}= \mathbb{E}_a(Z^{m_a}_{a}), \ \ f_a=\frac{Z^{m_a}_{a}}{\mathbb{E}_a(Z^{m_a}_{a})}$$
in perfect analogy with the path outlined in  \cite{g3}. We are then ready to state the following
\begin{proposition}\label{n-broken}
Let us consider the functional $\varphi(n,t)=n^{-1}\log \mathbb{E}_0(Z^n_0)$. The following relation holds
\be\frac{d}{dt}\varphi(n,t)=\frac{1}{2}{\langle S_{ii}-\widehat{S}^K_{ii}\rangle }^n_K
+\frac{1}{2}\sum^K_{a=0}(m_{a+1}-m_a)_n{\langle  S_{ij}-\widehat{S}^a_{ij}\rangle }^n_a\ee
where $\widehat{S}^0_{ij}=0,\ \ \ \ \ \widehat{S}^a_{ij}=\sum^a_{b=1}\widetilde{S}^b_{ij}$.
\end{proposition}

\subsection{Upper Bound and Parisi solution}

We can apply Proposition \ref{n-broken} to the interpolant $Z_K\equiv Z_t \equiv Z_N(\beta,t,x)$, where
\be\nonumber Z_N(\beta,t,x)=\sum_{\sigma_1...\sigma_N}\exp\Big(\beta\sqrt{\frac{N}{2}}K(\sigma)
+\beta\sqrt{1-t}\sum^K_{a=1}\sqrt{q_a-q_{a-1}}J^a_i\sigma_i)\Big)e^{\beta\ h\sum_{i}\sigma_i}\ee
and the $J^a_i$ are defined as the $B^a_i$ (see eq.(\ref{35}) and above) and $x_n$ mirrors the broken replica steps, namely we introduce a convex space $\chi_n$ whose elements are the
$x_n(q)$ piecewise
functions $x_n:q\rightarrow[n,1]$ such that $x_n(q)=m_a(n)$ for
$q_{a-1}<q\leq q_a\ \ \ \ \forall a=1,...,K$, with the prescription $q_0=0,\ q_K=1$.
\newline
Note that in this sense we wrote $Z_N(\beta,t,x)$ even though there is no explicit dependence on $x$ at the r.h.s.
\newline
We then consider the functional
\be\label{varfi} \varphi(n,t)= (Nn)^{-1}\log \mathbb{E}_0(Z^n_0) \ee
and introduce the following
\begin{lemma}
$$\varphi(n,1)=\varphi_N(n,\beta,h), \ \ \varphi(n,0)=\log2+f(0,h;x_n,\beta),$$
where $f$ satisfies the Parisi equation with $x_n$ as introduced in Section $2$.
\end{lemma}
Consequently the following Theorem holds
\begin{theorem} $\forall n \in (0,1]$ the functional $n$-quenched free energy $\varphi(n,t)$ defined in eq.(\ref{varfi}) respects the bound
$$\varphi(n,1)=\varphi_N(n,\beta,h)\leq\log2+f(0,h;x_n,\beta)-\frac{\beta^2}{4}\Big(1-\sum^K_{a=0}(m_{a+1}-m_a)_n q^2_a \Big)$$
uniformly in N.
\end{theorem}
\begin{proof} We can use Proposition \ref{n-broken}, keeping in mind the relations
\begin{eqnarray}
&& \mathbb{E}\Big(\beta^2\frac{N}{2}K(\sigma)K(\sigma')\Big) = \beta^2\frac{N}{2}q^2_{12}=S_{ij}, \\
&& \mathbb{E}\Big(\beta^2\sqrt{q_a-q_{a-1}}\sqrt{q_b-q_{b-1}}\sum_{i}J^a_i\sigma_i\sum_{j}J^a_j\sigma_j\Big), \nonumber
= \beta^2 N(q_a-q_{a-1})q_{12}=\widetilde{S}^a_{ij}.
\end{eqnarray}
to get
$$\frac{d}{dt}\varphi(n,t)=-\frac{\beta^2}{4}-\frac{\beta^2}{4}\sum^K_{a=0}(m_{a+1}-m_a)_n
{\langle  q^2_{12}-2q_a q_{12}\rangle }^n_a.$$
Filling with $q^2$ the square at the r.h.s. we obtain
$$\frac{d}{dt}\varphi(n,t)=-\frac{\beta^2}{4}(1-\sum^K_{a=0}(m_{a+1}-m_a)_n q^2_a) -\frac{\beta^2}{4}\sum^K_{a=0}(m_{a+1}-m_a)_n
{\langle  (q_{12}-q_a)^2\rangle }^n_a.$$
Lastly, it is enough to remember that
$$(m_{a+1}-m_{a})_n\geq 0\ \ \ \ \forall a=0,...,K\Rightarrow \varphi(n,1)\leq\varphi(n,0)-\frac{\beta^2}{4}(1-\sum^K_{a=0}(m_{a+1}-m_a)_n q^2_a),$$
to get the thesis.
\end{proof}
We can then define
\be\label{38} \alpha_P(\beta,h,x_n)=\log 2+n\frac{\beta^2}{4}+f(0,y;x_n,\beta)\mid_{y=h}-\frac{\beta^2}{2}\int^1_0qx_n(q)dq,\ee
and write furthermore that
$$\frac{1}{2}(1-\sum^K_{a=0}(m_{a+1}-m_a)_n q^2_a)=\int^1_0qx_n(q)dq-\frac{n}{2}$$
to state the next
\begin{theorem} The following bounds hold
\begin{eqnarray}\nonumber
&\lim_{N \to\infty}&\varphi_N(n,\beta,h)=\varphi(n,\beta,h)\leq\alpha_P(\beta,h,x_n)\Rightarrow \varphi(n,\beta,h)\leq\inf_{x_n}\alpha_P(\beta,h,x_n), \\
&\lim_{n \to 0}&\varphi(n,\beta,h)\leq\lim_{n \to 0}\inf_{x_n}\alpha_P(\beta,h,x_n)=\alpha_P(\beta,h),
\end{eqnarray}
and clearly  $\lim_{n \to 0}\alpha_P(\beta,h,x_n)=\alpha_P(\beta,h,x)$.
\end{theorem}

\subsection{The temperature of the disorder}

In this section we  want to try to emphasize the formal analogy between the "real" temperature $\beta$ and
an "effective" temperature $n$ as
\begin{eqnarray}
f(\beta) &=& \frac{1}{\beta}\mathbb{E}\log\sum_{\sigma}e^{-\beta H(\sigma;J)},\\
f(n) &=& \frac1n  \log \mathbb{E} e^{n\log Z(J)}.
\end{eqnarray}
Interestingly for a connection to the dynamical properties of glasses  \cite{leticia1} \cite{leticia2} \cite{peter1} \cite{peter2}, while the Boltzmann temperature $\beta$ rules the overall energy fluctuations of the system, $n$ seems to tackle the behavior inside the valleys of free energy themselves.
\newline
As  we are interested in thinking at $n$ as an effective temperature selecting valleys of free energies, we stress that by applying the framework we exploited so far, for $n=1$, $\chi_n$ collapses into the space of the constant unitary functions and the solution of eq. (\ref{38}) coincides with the annealed.
\newline
We know (see for instance  \cite{bipartiti}) that mean field spin systems often obey convex representations (trough their order parameters) in temperature. Still bridging, we note that
$$
\chi_n \ni  x_n : q \to [n,1] \Rightarrow \forall x_n \in \chi_n: \exists x_0 \in \chi_0 : x_n = n x_1+ (1-n)x_0(q).
$$
So we see that the space $\chi_n$ admits an analogous convex decomposition, with $n$ instead of $\beta$: $\chi_n = n \chi_1 \bigoplus (1-n)\chi_0$ \footnote{Strictly speaking, in the paper \cite{bipartiti} it was shown how to obtain such a decomposition for the free energies. Of course we can expand them in their irreducible overlap correlation functions so to carry on the mapping even at the level of order parameters.}.
\newline
To deepen this point we revise here the powerful approach investigated by Sherrington, Coolen and coworkers in a series of papers \cite{SC1,SC2,SC3}: At first, let us introduce the average $\mathbb{E}_{\sigma}$ of the configurations as
$$Z(\beta,J)=\frac{1}{2^{N}}\sum_{\sigma}e^{-\beta H(J,\sigma)}=\mathbb{E}_{\sigma}e^{-\beta H(J,\sigma)},$$
by which, annealed and quenched free energies can be written as
\begin{eqnarray}\label{annealed}
f_A(\beta)&=&-\frac{1}{\beta N}\log \mathbb{E}_J( Z(\beta,J)=-\frac{1}{\beta N}\log \mathbb{E}_J\mathbb{E}_{\sigma}e^{-\beta H(J,\sigma)},\\ \label{quenched}
f_Q(\beta)&=&-\frac{1}{\beta N}\mathbb{E}_J\log Z(\beta,J)=-\frac{1}{\beta N} \mathbb{E}_J\log \mathbb{E}_{\sigma}e^{-\beta H(J,\sigma)},
\end{eqnarray}
where $p(J)$ should not be confused with the a-priori $J$-distribution that is included in $\mathbb{E}_J$, and such that in the annealed case both the r.v. $J$ and $\sigma$ are thermalized  on the same timescale (related to $\beta$), while in the quenched case the r.v. $J$ is averaged after taking the logarithm, such that its dynamics is completely frozen w.r.t. the dynamics of the fast variables $\sigma$. As, so far, we used $n$ as a real interpolating parameter, we  want to see here if and how it can be though of as a quencher for the $J$.
\newline
To this task let us consider (implicitly defining it) the extended extensive free energy Boltzmann functional
\be\mathcal{H}=\mathbb{E}_J\mathbb{E}_{\sigma}p(J,\sigma)\Big(H(J,\sigma)+\frac{1}{\beta}\log p(J,\sigma)\Big)\ee
where $p(J,\sigma)$ is a properly introduced weight whose explicit expression we want to work out.
\newline
We restrict ourselves in searching for explicit expressions that allow the following decomposition
$$p(J,\sigma)=p(J)p(\sigma|J),$$
such that, by direct substitution we can write
\be \mathcal{H}=\mathbb{E}_Jp(J)\Big(H_{eff}(J)+\frac{1}{\beta}\log p(J)\Big)\ee
where $H_{eff}(J)$ is the standard extensive free energy\footnote{We allow ourselves in a little abuse of notation forgetting the $\beta$ dependence for now.} as
\be H_{eff}(J)=\mathbb{E}_{\sigma}p(\sigma|J)\Big(H(J,\sigma)+\frac{1}{\beta}\log p(\sigma|J)\Big).\ee
 Now, at fixed $J$, we can minimize  $H_{eff}(J)$ w.r.t. $p(\sigma|J)$ with the constraint  $\mathbb{E}_{\sigma}p(\sigma)|J)=1$ so to obtain the classical expression
$$p(\sigma|J)\equiv p(\sigma|J,\beta)=\frac{1}{Z(\beta,J)}e^{-\beta H(J,\sigma)},$$
where $Z(\beta,J)=\mathbb{E}_{\sigma}e^{-\beta H(J,\sigma)}$ is the standard partition function
and the extensive free energy assumes the familiar representation
\be H_{eff}(J)\equiv H_{eff}(J,\beta)=-\frac{1}{\beta}\log Z(\beta,J).\ee
\newline
Now let us instead minimize $\mathcal{H}$ w.r.t. $p(J)$ with two constraints:
the former being the normalization over $P(J)$, i.e. $\mathbb{E}_Jp(J)=1$, the latter being the choice of the entropy for the $J$ variables, which we retain in the classical equilibrium form even for these variables (implicitly assuming adiabaticity as in the seminal papers by Coolen)
$$-\frac{1}{\beta}\mathbb{E}_J p(J)\log p(J)=S(n,\beta).$$
Note that here we emphasize the $n$-dependence introduced in this further "entropy" due to the complexity of the choice of the $J$-distribution\footnote{Of course for simple systems, as for instance the Curie-Weiss model where $P(J)\sim \delta(J-1)$, this term does not contribute to thermodynamics and there is no $n$-dependence.}. Note further that this entropy is tuned by $\beta$.
\newline
Let us use $\lambda$ and $\mu$ for the Lagrange multipliers, such that the functional to be minimized can be read off as
\be \mathcal{H}+ \mu \large(\mathbb{E}_Jp(J)-1 \large)+\lambda \large( \frac{1}{\beta}\mathbb{E}_J p(J)\log p(J)+S(n,\beta) \large).\ee
By minimizing w.r.t. $p(J)$ we get
\be H_{eff}(J,\beta)+(\frac{\lambda+1}{\beta})+(\frac{\lambda+1}{\beta})\log p(J)+\mu=0 \ee
or simply
$$p(J)=e^{-\frac{\beta}{\lambda+1}H_{eff}(J)}e^{-\frac{\beta}{\lambda+1}\mu}.$$
Using the constraint over the normalization (the one ruled by $\mu$) we get immediately
$$e^{\frac{\beta}{\lambda+1}\mu}=\mathbb{E}_Je^{-\frac{\beta}{\lambda+1}H_{eff}(J)}.$$
We are left with the determination of $\lambda$: To this task we can always choose the function
$S(n,\beta)$ such that $\frac{1}{\lambda+1}=n$, so to get
\be p(J)\equiv p(J,\beta,n)=\frac{1}{\widetilde{Z}(\beta,n)}e^{-\beta nH_{eff}(J,\beta)},\ee
where
$$\widetilde{Z}(\beta,n)=\mathbb{E}_J e^{-\beta n H_{eff}(J,\beta)}.$$
The explicit expression defining $S(n,\beta)$  becomes
\be S(n,\beta)=-\frac{1}{\beta}\mathbb{E}_J p(J,\beta,n)\log p(J,\beta,n),\ee
such that, pasting the whole together, we get the explicit expression for the functional $\mathcal{H}(\beta,n)$, namely sharply the replica-trick free energy:
\be \mathcal{H}(\beta,n)=-\frac{1}{\beta n}\log\widetilde{Z}(\beta,n)=-\frac{1}{\beta n}\log \mathbb{E}_J\Big(Z(\beta,J)^n\Big). \ee
It is straightforward to check that, for instance, when considering the Curie-Weiss model, the $n$-dependance disappears, while it assumes the classical meaning when dealing with the Sherrington-Kirkpatrick one (e.g. equations (\ref{annealed}) and (\ref{quenched})).

\section{The commutativity of $n \to 0$ and $N \to \infty$}

Let us now extend the interpolation to tackle two i.i.d. copies of the original Hamiltonian  $H_{1},H_2$ as
\begin{equation}
H_N(\sigma,t)=\sqrt{t}H_1(\sigma)+\sqrt{1-t}H_2(\sigma),
\end{equation}
where we omitted the $N$-dependence in $H_1,H_2$ for the sake of clearness.
\newline
We can define the corresponding partition function as
\begin{equation}
Z(\beta,t)=\sum_{\sigma}e^{-\beta H(\sigma,t)},
\end{equation}
and define the interpolating functional as
\begin{equation}
\psi(n,t)=\frac{1}{n}\log \mathbb{E}_1(\exp (n\mathbb{E}_2(\log Z(\beta,t))
\end{equation}
where $ \mathbb{E}_{1,2}$ averages respectively over the disorders of  $ H_{1,2}$.
\newline
It is straightforward to check that
\begin{eqnarray}
\psi(n,1) &=& \frac{1}{n}\log \mathbb{E}_1(\exp (n\log Z(\beta,t=1))\equiv\frac{1}{n}\log \mathbb{E}(\exp (n\log Z(\beta))),\\
\psi(n,0) &=& \mathbb{E}_2(\log Z(\beta,t=0))\equiv\mathbb{E}(\log Z(\beta)),
\end{eqnarray}
where $Z(\beta)$ is the partition function of the original Hamiltonian.
\begin{proposition}\label{Propnuova}
After introducing
\begin{equation}
G(n,t)=\exp (n\mathbb{E}_2(\log Z(\beta,t)),
\end{equation}
and the $t$-dependent Boltzmann weights as $p(\sigma,t)=e^{-\beta H(\sigma,t)}/Z(\beta,t)$, the streaming of the functional $\psi(n,t)$  with respect to the interpolating parameter is
\begin{equation}
\frac{d\psi(n,t)}{dt}=n\frac{\beta^2}{2}\frac{1}{\mathbb{E}_1(G(n,t))}\mathbb{E}_1\Big(G(n,t)
\sum_{\sigma,\tau}\mathcal{C}(\sigma,\tau)\mathbb{E}_2(p(\sigma,t))\mathbb{E}_2(p(\tau,t)\Big).
\end{equation}
\end{proposition}
\begin{proof}
By a direct evaluation we get
$$\frac{d\psi(n,t)}{dt}=\frac{\mathbb{E}_1\Big(G(n,t)\mathbb{E}_2(\frac{d Z(\beta,t)}{dt}\frac{1}{Z(\beta,t)})\Big)}{\mathbb{E}_1(G(n,t))},$$
where
$$\frac{d Z(\beta,t)}{dt}=-\frac{\beta}{2}\sum_{\sigma}\Big(\frac{1}{\sqrt{t}}H_1(\sigma)-\frac{1}{\sqrt{1-t}}H_2(\sigma)\Big)e^{-\beta H(\sigma,t)}.$$
Then we write
$$\frac{d\psi(n,t)}{dt}=-\frac{\beta}{2}\frac{1}{\mathbb{E}_1(G(n,t))}(A-B),$$
where
\begin{eqnarray}
A&=&\mathbb{E}_1\Big(G(n,t)\mathbb{E}_2\sum_{\sigma}\Big(\frac{1}{\sqrt{t}}H_1(\sigma)p(\sigma,t)\Big)\Big),\\
B&=&\mathbb{E}_1\Big(G(n,t)\mathbb{E}_2\sum_{\sigma}\Big(\frac{1}{\sqrt{1-t}}H_2(\sigma)p(\sigma,t)\Big)\Big).
\end{eqnarray}
Introducing here the label $\tau$ with the usual meaning of another set of Ising spins $\tau_i = \pm 1, \ i \in (1,...,N)$, by applying Wick theorem to $A$ (on the family of random $H_1(\sigma)$) and calling the covariance matrix of $H_1(\sigma)$ $\mathcal{C}(\sigma,\tau)$ we get
\begin{eqnarray}
A&=&\frac{1}{\sqrt{t}}\sum_{\sigma}\mathbb{E}_1\Big(H_1(\sigma)G(n,t)\mathbb{E}_2(p(\sigma,t))\Big)\\
&=&\frac{1}{\sqrt{t}}\sum_{\sigma,\tau}\mathcal{C}(\sigma,\tau)\mathbb{E}_1\Big(\frac{\partial G(n,t)}{\partial H_1(\tau)}\mathbb{E}_2(p(\sigma,t))\Big)+G(n,t)\mathbb{E}_2(\frac{\partial p(\sigma,t)}{\partial H_1(\tau)})\Big).
\end{eqnarray}
We must then evaluate explicitly
$$\frac{\partial G(n,t)}{\partial H_1(\tau)}=-n\beta\sqrt{t}G(n,t)\mathbb{E}_2\Big(e^{-\beta H(\tau,t)}\frac{1}{Z(\beta,t)}\Big)=-n\beta\sqrt{t}G(n,t)\mathbb{E}_2\Big(p(\tau,t)\Big),$$
and
$$ \frac{\partial p(\sigma,t)}{\partial H_1(\tau)}=-\beta\sqrt{t}\Big(\delta_{\sigma\tau} p(\sigma,t)+p(\sigma,t)p(\tau,t)\Big).$$
Overall we can write
$$A=-\beta\mathbb{E}_1\Big(G(n,t)\sum_{\sigma,\tau}\mathcal{C}(\sigma,\tau)\Big[n\mathbb{E}_2(p(\sigma,t))\mathbb{E}_2(p(\tau,t)+
\mathbb{E}_2(\delta_{\sigma\tau} p(\sigma,t)+p(\tau,t)\Big]\Big).$$
By applying Wick theorem to $B$ (on the family of random $H_2(\sigma)$) and calling again its covariance matrix  $\mathcal{C}(\sigma,\tau)$ (as the two Hamiltonian are i.i.d.) we get
\begin{eqnarray}
B &=& \mathbb{E}_1 \Big(G(n,t)\mathbb{E}_2\sum_{\sigma}\Big(\frac{1}{\sqrt{1-t}}H_2(\sigma)p(\sigma,t)\Big)\Big) \\
&=& \frac{1}{\sqrt{1-t}}\mathbb{E}_1\Big( G(n,t)\sum_{\sigma,\tau}\mathcal{C}(\sigma,\tau)\mathbb{E}_2(\frac{\partial p(\sigma,t)}{\partial H_2(\tau)})\Big).
\end{eqnarray}
Mirroring the previous calculations, we get
$$ \frac{\partial p(\sigma,t)}{\partial H_2(\tau)}=-\beta\sqrt{1-t}\Big(\delta_{\sigma\tau} p(\sigma,t)+p(\sigma,t)p(\tau,t)\Big).$$
Pasting all together we get the thesis.
\end{proof}
\begin{remark}
The proposition still holds even if we consider  an external field coupled to the system and not only for $n \in [0,1]$.
\end{remark}
We are ready to state the next
\begin{theorem}\label{nuovo}
Let us recall that the SK-model is thermodynamically stable \cite{contucci}, namely it exists a constant $C < \infty$ such that $\lim_{N\to\infty}(1/N) \mathcal{C}(\sigma,\sigma) \leq C$, (and, as a consequence of the Schwartz inequality, $\lim_{N\to\infty}(1/N) \mathcal{C}(\sigma,\tau) \leq C$), and that it admits a sensible thermodynamic limit \cite{limterm}, then
$$\lim_{n \to 0^{+}}\lim_{N \to \infty}\frac{1}{N}\varphi_N(\beta,n)=\alpha(\beta).$$
\end{theorem}
\begin{proof}
It is immediate to check that $\varphi_N(\beta,n)$ is increasing in  $n$ for  $n\in [0,1]$ and this monotony is preserved in the thermodynamic limit, so that
\begin{eqnarray}
\exists \lim_{n \to 0^{+}}&\lim_{N \to \infty}&\frac{1}{N}\varphi_N(\beta,n),\\
&\lim_{N \to \infty}&\frac{1}{N}\varphi_N(\beta,n)\geq \lim_{N \to \infty}\frac{1}{N}\alpha_N(\beta)
=\alpha(\beta),
\end{eqnarray}
or simply
$$\lim_{n \to 0^{+}}\lim_{N \to \infty}\frac{1}{N}\varphi_N(\beta,n)\geq\alpha(\beta).$$
To proof the inverse inequality we use Proposition \ref{Propnuova}.
\newline
Let us consider
$$\psi_N(n,\beta,t)=\frac{1}{Nn}\log \mathbb{E}_1 \exp (n\mathbb{E}_2(\log Z_N(\beta,t))).$$
Of course we have that
\begin{eqnarray}
\psi_N(n,\beta,1)&=&\varphi_N(\beta,n),\\
\psi_N(n,\beta,0)&=&\alpha_N(\beta),
\end{eqnarray}
and we can write
$$\psi_N(n,\beta,1)-\psi_N(n,\beta,0)=\int_0^1 dt\frac{\partial}{\partial t} \psi_N(n,\beta,t),$$
where
\begin{eqnarray}
&& \frac{\partial}{\partial t} \psi_N(n,\beta,t) = \\ \nonumber
&& \frac{n}{N}\frac{\beta^2}{2}\frac{1}{\mathbb{E}_1(G_N(n,\beta,t))}\mathbb{E}_1\Big(G_N(n,\beta,t)
\sum_{\sigma,\tau}\mathcal{C}_N(\sigma,\tau)\mathbb{E}_2(p_N(\sigma,\beta,t))\mathbb{E}_2(p_N(\tau,\beta,t)\Big).
\end{eqnarray}
Bounding $\mathcal{C}_N(\sigma,\tau)$ with is sup and noticing that
$$\sum_{\sigma,\tau}\mathbb{E}_2(p_N(\sigma,\beta,t))\mathbb{E}_2(p_N(\tau,\beta,t))=1,
$$ we have that
$$\frac{\partial}{\partial t} \psi_N(n,\beta,t)\leq \frac{n}{N}\frac{\beta^2}{2}\max_{\sigma,\tau}\mathcal{C}_N(\sigma,\tau).$$
We can use now the property of thermodynamic stability to obtain
$$\lim_{N \to \infty}\frac{1}{N}\varphi_N(\beta,n)-\lim_{N \to \infty}\frac{1}{N}\alpha_N(\beta)
\leq n\frac{\beta^2}{2}C,$$
or simply
$$\lim_{n \to 0^{+}}\lim_{N \to \infty}\frac{1}{N}\varphi_N(\beta,n)-\alpha(\beta)\leq 0,$$
which is the inverse bound.
\newline
For the commutativity of $\lim_n$ and $\lim_N$ now it is enough to prove the inverse limit. This can be achieved immediately by applying De l'Hopital Theorem to  $\varphi_N(\beta,n)$ in $n$ to get
$$ \lim_{n \to 0^{+}}\varphi_N(\beta,n)=\alpha_N(\beta),$$
such that
$$\lim_{N \to \infty}\frac{1}{N}\lim_{n \to 0^{+}}\varphi_N(\beta,n)=\alpha(\beta).$$
\end{proof}
\begin{remark}
We stress  that, despite in this paper we limit ourselves to the investigation of the properties of the pure SK model, the methods exploited in this section apply to a broad range of models, as discussed for instance in \cite{contucci}.
\end{remark}

At the end we enlarge the scheme introduced in this section by defining the following functional
\begin{equation}
\psi(n,m,t):=\frac{1}{n}\log \mathbb{E}_1\Big(\exp \Big[\frac{n}{m}\log\mathbb{E}_2(\exp (m\log Z(t)\Big]\Big),
\end{equation}
where, as usual, $ \mathbb{E}_{1,2}$ average over the disorder respectively $ H_{1,2}$.
\newline
Again it is straightforward to check that
\begin{eqnarray}
\psi(n,m,1)&=&\frac{1}{n}\log \mathbb{E}_1(\exp (n\log Z(1))\equiv\frac{1}{n}\log \mathbb{E}(\exp (n\log Z))\\
\psi(n,m,0)&=&\frac{1}{m}\log \mathbb{E}_2(\exp (m\log Z(0))\equiv\frac{1}{m}\log \mathbb{E}(\exp (m\log Z))
\end{eqnarray}
and that the following generalization of Proposition \ref{Propnuova} holds
\begin{eqnarray} &&\frac{d\psi(n,m,t)}{dt}= \\ \nonumber
&& \frac{\beta^2}{2}\frac{(n-m)}{\mathbb{E}_1(G(n,m,t))}\mathbb{E}_1\Big(G(n,m,t)
\sum_{\sigma,\tau}\mathcal{C}(\sigma,\tau)\mathbb{E}_2(p(\sigma,t)b(m,t))\mathbb{E}_2(p(\tau,t)b(m,t)\Big),
\end{eqnarray}
where
\begin{equation}
G(n,m,t):=\exp \Big[\frac{n}{m}\log\mathbb{E}_2\Big(\exp(m\log Z(t))\Big)\Big],
\end{equation}
by which we can argue that the $n$-quenched free energy  $\varphi_N(\beta,n)$ has Lipschitz constant equal to $L=C \beta^2/2$.

\section*{Acknowledgment}

The strategy outlined in this research article belongs to the study supported by the
Italian Ministry for Education and Research, FIRB grant number $RBFR08EKEV$,  and partially by Sapienza Universit\`a di Roma.
\newline
FG is partially funded by INFN (Istituto Nazionale di Fisica Nucleare) which is also acknowledged.
\newline
The authors are pleased to thank Dmitry Panchenko and an anonymous referee for their kind suggestions.

\label{lastpage}

\end{document}